\documentclass[reqno]{amsart}
\usepackage{mathrsfs}
\usepackage{cite}
\newtheorem{theorem}{Theorem}

\newtheorem{remark}{Remark}
\newtheorem{proposition}{Proposition}
\newtheorem{corollary}{Corollary}
\begin{document}

\title[Additional reductions in the  k-constrained modified KP hierarchy]
{Additional reductions in the k-constrained modified KP hierarchy}%
\author{O. Chvartatskyi}
\author{Yu. Sydorenko}
\address{Ivan Franko National University of L'viv, 1 Universytetska st., 79000 L'viv, Ukraine}
\email{alex.chvartatskyy@gmail.com, y\_sydorenko@franko.lviv.ua}%
UDC 517.9

MSC 35Q51,  35Q53, 35Q55, 37K10, 37K40, 37K35, 35C08


\begin{abstract}
Additional reductions in the modified k-constrained KP hierarchy are
proposed. As a result we obtain generalizations of Kaup-Broer
system, Korteweg-de Vries equation and a modification of Korteweg-de
Vries equation that belongs to the modified k-constrained KP
hierarchy. We also propose solution generating technique based on
binary Darboux transformations for the obtained equations.
\end{abstract}
\maketitle

\section{Introduction}

The  algebraic  constructions  of the  well-known  Kyoto group
\cite{Ohta}, which  are  called  the Sato theory, play an important
role in the contemporary  theory of nonlinear integrable systems of
mathematical  and theoretical physics.  The leading place in  these
investigations  is occupied by the theory of equations of
Kadomtsev-Petviashvili type (KP hierarchy)  and  their
generalizations and applications  \cite{Ohta,LDA,Blas}.

One of known generalizations of the KP hierarchy arise as a result
of k-symmetry constraints (so-called k-cKP hierarchy) that were
investigated in \cite{SS,KSS,Chenga1,CY,Chenga2}. k-cKP hierarchy
are closely connected with so-called KP equation with
self-consistent sources (KPSCS) \cite{M1,M2,M3,Sam}. 
Multicomponent k-constraints of the KP hierarchy were introduced in
\cite{SSq} and investigated in
\cite{Oevel93,ZC,Oevel96,Aratyn97,WLG}. This extension of k-cKP
hierarchy contains vector (multicomponent) generalizations of such
physically relevant systems like the nonlinear Schr\"odinger
equation, the Yajima-Oikawa system, a generalization of the
Boussinesq equation, and the Melnikov system.


 The modified k-constrained KP (k-cmKP) hierarchy was proposed in
\cite{KSO,OC}. It contains, for example, the vector Chen-Lee-Liu,
the modified KdV (mKdV) equation and their multi-component
extensions. The k-cmKP hierarchy and dressing methods for it via
integral transformations were investigated in \cite{SydA,BS1,PHD}

In \cite{MSS,6SSS} (2+1)-dimensional extensions of the k-cKP
hierarchy ((2+1)-dimensional k-cKP hierarchy) were introduced and
dressing methods via differential transformations were investigated.
Some systems of this hierarchy were investigated via binary Darboux
transformations in \cite{BS1,PHD}. This hierarchy was also
rediscovered recently in \cite{LZL1,LZL2}. Matrix generalizations of
(2+1)-dimensional k-cKP hierarchy were considered in
\cite{Zeng,YYQB}.

In this paper our aim is to consider additional reductions of the
k-cmKP hierarchy that lead to new generalizations of well-known
integrable systems. We also investigated dressing methods for the
obtained systems via integral transformations that arise from Binary
Darboux Transformations (BDT).

 This work is organized as follows. In Section 2 we
 present a short survey of results on constraints for the KP hierarchies
 including the k-cmKP hierarchy. In Section 3 we investigate Lax
representations obtained as a result of additional reductions in the
k-cmKP hierarchy and corresponding nonlinear systems. Section 4
presents results on dressing methods for Lax pairs obtained in
Section 3. In the final section, we discuss the obtained results and
mention problems for further investigations.

\section{Symmetry constraints of the KP hierarchy}\label{kckp}

Let us recall some basic objects and notations concerning KP
hierarchy, modified KP hierarchy, their multicomponent k-constraints
and their (2+1)-extensions. A Lax representation of the KP hierarchy
is given by
\begin{equation}\label{ssk}
L_{t_n}=[B_n,L],\,    \qquad n\geq1,
\end{equation}
where $L=D+U_1D^{-1}+U_2D^{-2}+\ldots$ is a scalar
pseudodifferential operator, $t_1:=x$, $D:=\frac{\partial}{\partial
x}$, and $B_n:= (L^n)_+
 := (L^n)_{\geq0}=D^n+\sum_{i=0}^{n-2}u_iD^i$ is
the differential operator part of $L^n$. The consistency condition
(zero-curvature equations), arising from the commutativity of flows
(\ref{ssk}), is

\begin{equation}\label{LP}
B_{n,t_k}-B_{k,t_n}+[B_n,B_k]=0.
\end{equation}

Let $B^{\tau}_n$ denote the formal transpose of $B_n$, i.e.
$B^{\tau}_n:=(-1)^nD^n+\sum_{i=0}^{n-2}(-1)^iD^iu^{\top}_i$, where
$^{\top}$ denotes the matrix transpose. We will use curly brackets
to denote the action of an operator on a function whereas, for
example, $B_n \, q$ means the composition of the operator $B_n$ and
the operator of multiplication by the function $q$. The following
formula holds for $B_nq$ and $B_n\{{q}\}$:
$B_n\{{q}\}:=(B_n{q})_{=0}=B_nq-(B_n{q})_{>0}.$ In the case $k=2$,
$n=3$ formula (\ref{LP}) presents a Lax pair for the
Kadomtsev-Petviahvili equation. Its Lax pair was obtained in
\cite{D} (see also \cite{Zakharov}).

The multicomponent k-constraints of the KP hierarchy is given by \cite{SSq} 
\begin{equation}\label{eq11}
  L_{t_n}=[B_n,L],
  \end{equation}
with the k-symmetry reduction
\begin{equation}\label{eq21}
    L_k:=L^k=B_k+\sum_{i=1}^m\sum_{j=1}^mq_im_{ij}D^{-1}r_j=B_k+{\bf
q}{\mathcal M}_0D^{-1}{\bf r}^{\top},
\end{equation}
where ${\bf q}=(q_1,\ldots,q_m)$ and ${\bf r}=(r_1,\ldots,r_m)$ are
vector functions, ${\mathcal M}_0=(m_{ij})_{i,j=1}^m$ is a constant
$m\times m$ matrix. In the scalar case ($m=1$) we obtain
k-constrained KP hierarchy \cite{SS,KSS,Chenga1,CY,Chenga2}. The
hierarchy given by (\ref{eq11})-(\ref{eq21}) admits the Lax
representation (here $k\in{\mathbf{N}}$ is fixed):
\begin{equation}\label{hier1}
  [L_k,M_n]=0,\,\,\, L_k=B_k+{\bf q}{\mathcal M}_0D^{-1}{\bf r}^{\top}, \quad
  M_n=\partial_{t_n}-B_n.
\end{equation}
Lax equation (\ref{hier1}) is equivalent to the following system:
\begin{equation}\label{hier231}
[L_k,M_n]_{\geq0}=0,\,\, M_n\{{\bf q}\}=0,\,\,\,M_n^{\tau}\{{{\bf
r}}\}=0.
\end{equation}

Below we will also use the formal adjoint
$B^*_n:=\bar{B}^{\tau}_n=(-1)^nD^n+\sum_{i=0}^{n-2}(-1)^iD^i{u}^*_i$
of $B_n$, where $^\ast$ denotes the Hermitian conjugation (complex
conjugation and transpose).

 For $k=1$, the hierarchy given by (\ref{hier231}) is a multi-component generalization of the AKNS
hierarchy. For $k=2$ and $k=3$, one obtains vector generalizations
of the Yajima-Oikawa and Melnikov \cite{M1,M2,M3} hierarchies,
respectively. An essential extension of the k-cKP hierarchy is its
(2+1)-dimensional generalization introduced in \cite{MSS,6SSS} and
rediscovered in \cite{LZL1,LZL2}.

In \cite{KSO,OC}, a k-constrained modified KP (k-cmKP) hierarchy was
introduced and investigated. Dressing methods for k-cmKP hierarchy
under additional $D$-Hermitian reductions were also investigated in
\cite{SydA,BS1}. At first we recall the definition of the modified
KP hierarchy.
A Lax representation of this hierarchy is given by
\begin{equation}\label{sskMod}
L_{t_n}=[B_n,L],\,    \qquad n\geq1,
\end{equation}
where $L=D+U_0+U_1D^{-1}+U_2D^{-2}+\ldots$ and $B_n:= (L^n)_{>0}
 :=D^n+\sum_{i=1}^{n-1}u_iD^i$ is
the purely differential operator part of $L^n$. The consistency
condition arising from the commutativity of flows (\ref{sskMod}), is
\begin{equation}\label{LP}
B_{n,t_k}-B_{k,t_n}+[B_n,B_k]=0.
\end{equation}
The multicomponent k-constraints of the modified KP hierarchy are given by the operator equation: 
\begin{equation}\label{eq1}
  L_{t_n}=[B_n,L],
  \end{equation}
with the k-symmetry reduction
\begin{equation}\label{eq2}
    L_k:=L^k=B_k-\sum_{i=1}^m\sum_{j=1}^mq_im_{ij}D^{-1}r_jD=B_k-{\bf
q}{\mathcal M}_0D^{-1}{\bf r}^{\top}D,
\end{equation}
where ${\bf q}=(q_1,\ldots,q_m)$ and ${\bf r}=(r_1,\ldots,r_m)$ are
vector functions, ${\mathcal M}_0=(m_{ij})_{i,j=1}^m$ is a constant
$m\times m$ matrix.
The hierarchy (\ref{eq1})-(\ref{eq2}) admits the Lax representation
(here $k\in{\mathbf{N}}$ is fixed):
\begin{equation}\label{hier}
\begin{array}{l}
  [L_k,M_n]=0,\,\,\, L_k=B_k-{\bf q}{\mathcal M}_0D^{-1}{\bf r}^{\top}D, \\
M_n=\alpha_n\partial_{t_n}-B_n,\,\,B_n=D^n+\sum_{i=1}^{n-1}u_iD^i.
\end{array}
\end{equation}
We can rewrite the Lax pair (\ref{hier}) in the following way:
\begin{equation}\label{hierxx}
\begin{array}{l}
  [L_k,M_n]=0,\,\,\, L_k=B_k-{\bf q}{\mathcal M}_0{\bf r}^{\top}+{\bf q}{\mathcal
M}_0D^{-1}{\bf r}_x^{\top},\\
  M_n=\alpha_n\partial_{t_n}-B_n,\,\,B_n=D^n+\sum_{i=1}^{n-1}u_iD^i.
\end{array}
\end{equation}
From Lax representation for k-cKP hierarchy
(\ref{hier1})-(\ref{hier231}) and representation (\ref{hierxx}) we
come to conclusion that equation $[L_k,M_n]=0$ in (\ref{hier}) is
equivalent to the following system: $ [L_k,M_n]_{>0}=0$, $M_n\{{\bf
q}\}=0$, $(M_n^{\tau})\{{{\bf r}}_x\}=0$ ($[L_k,M_n]_{=0}=0$ since
$[L_k,M_n]\{1\}=0$). We can rewrite the last equation in the following form:
$(D^{-1}M_n^{\tau}D)\{{{\bf r}}\}=0$ to keep the order of
differentiation equal to $n$. As a result we obtain:
\begin{equation}\label{hier23}
[L_k,M_n]_{>0}=0,\,\, M_n\{{\bf
q}\}=0,\,\,\,(D^{-1}M_n^{\tau}D)\{{{\bf r}}\}=0.
\end{equation}
The hierarchy (\ref{hier}) contains vector generalizations of the
Chen-Lee-Liu ($k=1$), the modified multi-component Yajima-Oikawa
($k=2$) and Melnikov ($k=3$) hierarchies. Consider some equations
that can be obtained from (\ref{hier}) under certain choice of $k$
and $n$ (see \cite{PHD}).
\begin{enumerate}
\item $k=1$, $n=2$.
Then (\ref{hier}) becomes
\begin{equation}\label{L1M2}
L_1=D-{{\bf q}}{\mathcal M}_0D^{-1}{{\bf r}}^{\top}D,\,\,
M_2=\alpha_2\partial_{t_2}-D^2+2{{{\bf q}}}{\mathcal M}_0{{{\bf
r}}}^{\top}D.
\end{equation}
In this case equation (\ref{hier23}) becomes the following system:
\begin{equation}\label{CLL}
\alpha_2{\bf q}_{t_2}-{\bf q}_{xx}+2{{{\bf q}}}{\mathcal M}_0{{{\bf
r}}}^{\top}{\bf q}_x=0,\,\alpha_2{\bf r}^{\top}_{t_2}+{\bf
r}^{\top}_{xx}+2{\bf r}^{\top}_x{{{\bf q}}}{\mathcal M}_0{{{\bf
r}}}^{\top}=0.
\end{equation}
Under additional Hermitian conjugation reduction: $\alpha_2=i$,
${\mathcal M}_0=-{\mathcal M}_0^*$, ${\bf r}^{\top}={{\bf q}}^*$
($L^*_1=-D^{-1}L_1D$, $M^*_2=D^{-1}M_2D$) in (\ref{CLL}) we obtain
the Chen-Lee-Liu equation:
\begin{equation}\label{CLLiu}
i{\bf q}_{t_2}-{\bf q}_{xx}+2{\bf q}{\mathcal M}_0{\bf q}^*{\bf
q}_x=0.
\end{equation}
\item $k=1$, $n=3$.
In this case (\ref{hier}) takes the form:
\begin{equation}\label{LM}
\begin{array}{l}
L_1=D-{{\bf q}}{\mathcal M}_0D^{-1}{{\bf r}}^{\top}D,\,\,\\
M_3=\alpha_3\partial_{t_3}-{D}^3+3{\bf q}{\mathcal M}_0{\bf
r}^{\top}{D}^2+3[{\bf q}_x{\mathcal M}_0{\bf r}^{\top}-({\bf
q}{\mathcal M}_0{\bf r}^{\top})^2]{D},
\end{array}
\end{equation}
and equations (\ref{hier23}) read:
\begin{equation}
\label{bi15}
\begin{array}{c}
\alpha_3{\bf q}_{t_3}={\bf q}_{xxx}-3({\bf q}{\mathcal M}_0{\bf
r}^{\top}){\bf q}_{xx}- 3({\bf q}_x{\mathcal M}_0{\bf r}^{\top}-
({\bf
q}{\mathcal M}_0{\bf r}^{\top})^2){\bf q}_x,\\
\alpha_3{\bf r}^{\top}_{t_3}={\bf r}^{\top}_{xxx}+3{\bf
r}^{\top}_{xx}({\bf q}{\mathcal M}_0{\bf r}^{\top})+ 3{\bf
r}^{\top}_x({\bf q}{\mathcal M}_0{\bf r}_x^{\top}+ ({\bf q}{\mathcal
M}_0{\bf r}^{\top})^2).
\end{array}
\end{equation}
After reduction of Hermitian conjugation: $\alpha_3=1$, ${\bf
r}^{\top}={\bf q}^*$, ${\mathcal M}_0=-{\mathcal M}_0^*$
($L_1^*=-D^{-1}L_1D$, $M_3^*=-D^{-1}M_3D$) (\ref{bi15}) becomes:
\begin{equation}\label{kkdv}
{\bf q}_{t_3}={\bf q}_{xxx}-3({\bf q}{\mathcal M}_0{\bf q}^{*}){\bf
q}_{xx}- 3({\bf q}_x{\mathcal M}_0{\bf q}^{*}- ({\bf q}{\mathcal
M}_0{\bf q}^{*})^2){\bf q}_x.
\end{equation}
\item $k=2$, $n=2$. After additional reduction in (\ref{hier}):
$\alpha_2=i$, $u_1:=iu$, $u=u(x,t_2)\in{\mathbb{R}}$,
${\mathcal{M}}_0={\mathcal{M}}_0^*$ Lax pair in (\ref{hier23})
reads:
\begin{equation}\nonumber
[L_2,M_2]=0,\,\, L_2={D}^2+iu{D}-{\bf q}{\mathcal M}_0{D}^{-1}{\bf
q}^*{D},\,\,M_2=i\partial_{t_2}-D^2-iuD,
\end{equation}
and equation (\ref{hier23}) becomes the modified Yajima-Oikawa
equation:
\begin{equation}\nonumber
\left.\begin{array}{l}
i{\bf q}_{t_2}={\bf q}_{xx}+iu{\bf q}_x,\quad
 u_{t_2}=2({\bf q}{\mathcal M}_0{\bf q}^*)_x.
\end{array}\right.
\end{equation}
\end{enumerate}
In the next section we will introduce additional reductions in
Chen-Lee-Liu hierarchy. As a result we will obtain generalizations
of the Kaup-Broer system, Korteweg-de Vries equation, modified
Korteweg-de Vries equation and their scalar coupled versions.

\section{Additional reductions in the modified k-constrained KP hierarchy}

For further convenience let us make a change in  formulae
(\ref{hier}):
\begin{equation}\label{forl}
{\bf q}\rightarrow{\tilde{\bf q}}, \,\,{\bf r}\rightarrow{\tilde{\bf
r}}, \,\,{\mathcal M}_0\rightarrow \tilde{{\mathcal M}_0}.
\end{equation}
After the change (\ref{forl}) the hierarchy (\ref{hier})  reads:
\begin{equation}\label{hierNew}
\begin{array}{l}
  [L_k,M_n]=0,\,\,\, L_k=B_k-{\tilde{\bf q}}{\tilde{\mathcal M}}_0D^{-1}{\tilde{\bf r}}^{\top}D, \quad
  M_n=\alpha_n\partial_{t_n}-B_n,\,\,\\B_n=D^n+\sum_{i=1}^{n-1}u_iD^i.
\end{array}
\end{equation}
Let us make the additional reduction in (\ref{hierNew}):
\begin{equation}\label{rd}
\begin{array}{l}
\tilde{\bf q}:=(q_1,\ldots,q_m,-v-\beta D^{-1}\{u\},1)=({\bf
q},-v-\beta D^{-1}\{u\},1),\,\,\,\\ {\tilde{\mathcal
M}}_0=\left(\!\begin{array}{ccc}{\mathcal M}_0&0&0\\0&1&0\\0&0&1
\end{array}\!\!
\right),\,\, \tilde{\bf r}:=(r_1,\ldots,r_m,1,\beta
D^{-1}\{u\})=({\bf r},1,\beta D^{-1}\{u\}),
\end{array}
\end{equation}
where ${\mathcal{M}}_0$ is $(m\times m)$-constant matrix, ${\bf q}$
and ${\bf r}$ are $m$-component vectors, $u$ and $v$ are scalar
functions, $\beta\in{\mathbb{R}}$, $D^{-1}\{u\}$ denotes indefinite
integral of the function $u$ with respect to $x$. After reduction
(\ref{rd}) k-cmKP hierarchy (\ref{hierNew}) takes the form:
\begin{equation}\label{NH}
\begin{array}{l}
  \!\![L_k,M_n]=0,\,L_k=B_k-{\bf q}{{\mathcal M}}_0D^{-1}{\bf r}^{\top}D+v+\beta D^{-1}u,\,
  M_n=\alpha_n\partial_{t_n}-B_n,\,\,\\B_n=D^n+\sum_{i=1}^{n-1}u_iD^i.
\end{array}
\end{equation}
In the following subsections we will investigate hierarchy
(\ref{NH}) in case $k=1$.
\subsection{Reductions of the Chen-Lee-Liu system}

Let us put $k=1$, $n=2$. Then Lax pair (\ref{NH}) becomes:
\begin{equation}\label{LM1}
\begin{array}{c}
[L_1,M_2]=0,\,\,L_1=D-{\bf q}{\mathcal M}_0D^{-1}{\bf
r}^{\top}D+\beta D^{-1}u+v,
\\
M_2=\alpha_2\partial_{t_2}-D^2+2({\bf q}{\mathcal M}_0{\bf
r}^{\top}-v)D.
\end{array}
\end{equation}
A system  that corresponds to equation (\ref{LM1}) has the form:
\begin{equation}\label{sysLM}
\begin{array}{c} \alpha_2 {\bf q}_{t_2}={\bf q}_{xx}-2({\bf
q}{\mathcal M}_0{\bf r}^{\top}-v){\bf q}_x,\quad \alpha_2{\bf
r}^{\top}_{t_2}=-{\bf r}^{\top}_{xx}-2{\bf r}^{\top}_x({\bf
q}{\mathcal M}_0{\bf r}^{\top}-v),\\
\alpha_2 u_{t_2}+u_{xx}+2\left(u({\bf q}{\mathcal M}_0{\bf
r}^{\top}-v)\right)_x=0,\\
-\alpha_2 v_{t_2}+2\beta u_x+v_{xx}-2\left({\bf q}{\mathcal M}_0{\bf
r}^{\top}-v\right)v_x=0.
\end{array}
\end{equation}

Consider additional reductions of Lax pair (\ref{LM1}) and system
(\ref{sysLM}).

\begin{enumerate}
\item Assume that ${\mathcal M}_0=-{\mathcal M}_0^*$, ${\bf
r}^{\top}={\bf q}^*$, $v=-2i{\rm Im}(\beta D^{-1}\{u\})$
($L^*_1=-DL_1D^{-1},\,$ $\,M^*_2=DM_2D^{-1}$). Then equation
(\ref{sysLM}) takes the form:
\begin{equation}\nonumber
\left.\begin{array}{c} \alpha_2 {\bf q}_{t_2}={\bf q}_{xx}-2(2i{\rm
Im}(\beta D^{-1}\{u\})+{\bf
q}{\mathcal M}_0{\bf q}^{*}){\bf q}_x,\\
\alpha_2 u_{t_2}+u_{xx}+2\left(u(2i{\rm Im}(\beta D^{-1}\{u\})+{\bf
q}{\mathcal M}_0{\bf
q}^{*})\right)_x=0.\\
\end{array}\right.
\end{equation}

\item Let us put ${\mathcal M}_0=0$ in operators $L_1$ and $M_2$:
$L_1=D+\beta D^{-1}u+v$, $M_2=\alpha_2\partial_{t_2}-D^2-2vD$. Then
equation (\ref{sysLM}) becomes the Kaup-Broer system:
\begin{equation}\label{Bg}
\begin{array}{c} \alpha_2 u_{t_2}+u_{xx}-2(uv)_x=0,\,\,
-\alpha_2 v_{t_2}+2\beta u_x+v_{xx}+2vv_x=0.
\end{array}
\end{equation}
In case $u=0$ in (\ref{Bg}) we obtain the Burgers equation:
$-\alpha_2 v_{t_2}+v_{xx}-vv_x=0$.

\item Consider the case $u=0$ in operators $L_1$ and $M_2$ (\ref{sysLM}):
$L_1=D-{\bf q}{\mathcal M}_0D^{-1}{\bf r}^{\top}D+v$,
$M_2=\alpha_2\partial_{t_2}-D^2+2(v+{\bf q}{\mathcal M}_0{\bf
r}^{\top})D$. Then (\ref{sysLM}) reads:
\begin{equation}\nonumber
\left.\begin{array}{c} \alpha_2 {\bf q}_{t_2}={\bf q}_{xx}-2({\bf
q}{\mathcal M}_0{\bf r}^{\top}-v){\bf q}_x,\,\,\,\,
\alpha_2{\bf r}^{\top}_{t_2}=-{\bf r}^{\top}_{xx}-2{\bf
r}^{\top}_x({\bf
q}{\mathcal M}_0{\bf r}^{\top}-v),\\
-\alpha_2 v_{t_2}+v_{xx}-\left({\bf q}{\mathcal M}_0{\bf
r}^{\top}-v\right)v_x=0.
\end{array}\right.
\end{equation}
\end{enumerate}

\subsection{Reductions of the modification of Korteweg-de Vries system (\ref{kkdv})}

Now let us consider the hierarchy (\ref{NH}) in case $k=1$, $n=3$.
Then its Lax pair $L_1$, $M_3$ in (\ref{NH}) reads:

\begin{equation}\label{LM}
\begin{array}{c}
[L_1,M_3]=0,\,\, L_1=D-{\bf q}{\mathcal M}_0D^{-1}{\bf
r}^{\top}D+\beta D^{-1}u+v,\,\,
M_3=\alpha_3\partial_{t_3}-D^3-\\-3(v-{\bf q}{\mathcal M}_0{\bf
r}^{\top})D^2-3\left(({\bf q}{\mathcal M}_0{\bf r}^{\top}-v)^2-{\bf
q}_x{\mathcal M}_0{\bf r}^{\top}+\beta u+v_x\right)D.
\end{array}
\end{equation}
Commutator equation in (\ref{LM}) is equivalent to the system:

\begin{equation}\label{S}
\left.
\begin{array}{c}
-\alpha_3
v_{t_3}+v_{xxx}+3vv_{xx}+3v^2v_x+3v_x^2+6\beta(uv)_x+\\+3\left\{({\bf
q}{\mathcal M}_0{\bf r}^{\top})^2-{\bf q}_x{\mathcal M}_0{\bf
r}^{\top}\right\}v_x-3{\bf q}{\mathcal M}_0{\bf
r}^{\top}v_{xx}-\\-6{\bf q}{\mathcal M}_0{\bf
r}^{\top}vv_x-3\beta({\bf q}{\mathcal M}_0{\bf
r}^{\top}u)_x-3\beta{\bf
q}{\mathcal M}_0{\bf r}^{\top}u_x=0,\\
\alpha_3 {\bf q}_{t_3}={\bf q}_{xxx}+3(v-{\bf q}{\mathcal M}_0{\bf
r}^{\top}){\bf q}_{xx}+\\+3\left\{({\bf q}{\mathcal M}_0{\bf
r}^{\top}-v)^2-{\bf q}_x{\mathcal M}_0{\bf r}^{\top}+v_x+\beta
u\right\}{\bf q}_x,\\
\alpha_3{\bf r}^{\top}_{t_3}={\bf r}^{\top}_{xxx}- 3\left({\bf
r}^{\top}_x\left(v-{\bf q}{\mathcal M}_0{\bf
r}^{\top}\right)\right)_x+\\+3{\bf r}^{\top}_x\left\{({\bf
q}{\mathcal
M}_0{\bf r}^{\top}-v)^2-{\bf q}_x{\mathcal M}_0{\bf r}^{\top}+v_x+\beta u\right\},\\
\alpha_3 u_{t_3}=u_{xxx}-3\left(u(v-{\bf q}{\mathcal M}_0{\bf
r}^{\top})\right)_{xx}+\\+3\left(u\left(({\bf q}{\mathcal M}_0{\bf
r}^{\top}-v)^2-{\bf q}_x{\mathcal M}_0{\bf r}^{\top}+v_x+\beta
u\right)\right)_x.
\end{array}
\right.
\end{equation}
Consider additional reductions in Lax pair (\ref{LM}) and
corresponding system (\ref{S}).

\begin{enumerate}
\item Assume that $v=-2i{\rm Im}(\beta D^{-1}\{u\})$, ${\bf q}^*={\bf
r}^{\top}$, $u\in{\mathbb{R}}$, ${\mathcal M}_0=-{\mathcal M}_0^*$
($L^*_1=-DL_1D^{-1},\,\,$ $M_3^*=-DM_3D^{-1}$). Then system
(\ref{S}) takes the form:
\begin{equation}\label{sys}
\left.\begin{array}{c} \alpha_3{\bf q}_{t_3}={\bf
q}_{xxx}-3(2i{\rm}{Im}(\beta u)+{\bf q}{\mathcal M}_0{\bf q}^*){\bf
q}_{xx}+\\+3\left(\left({\bf q}{\mathcal M}_0{\bf q}^*+2i{\rm
Im}(\beta u)\right)^2\right.\left.-{\bf q}_x{\mathcal M}_0{\bf
q}^*+\beta
u-2i {\rm Im}(\beta u)\right){\bf q}_x,\\
\alpha_3 u_{t_3}=u_{xxx}+3\left\{u\left(2i{\rm Im}(\beta u)+{\bf
q}{\mathcal M}_0{\bf
q}^*\right)\right\}_{xx}+\\+3\left(u\left\{({\bf{q}}{\mathcal
M}_0{\bf q}^*+2i{\rm Im}(\beta u))^2-\right.\right.{\bf
q}_x{\mathcal M}_0{\bf q}^*+\left.\left.\beta u-2i{\rm Im}(\beta
u)\right\}\right)_x.
\end{array}\right.
\end{equation}
\begin{enumerate}
\item Let us assume that in addition to reductions described in item 1 functions  ${\bf q}$
and $u$ with matrix ${\mathcal M}_0$ are real-valued (i.e., matrix ${\mathcal M}_0$ is skew-symmetric: ${\mathcal{M}}_0^{\top}=-{\mathcal{M}}_0$)  and $v=0$. Then
the scalar ${\bf q}{\mathcal M}_0{\bf q}^{\top}=0$ since ${{\bf
q}{\mathcal M}_0{\bf q}^{\top}}=-({{\bf q}{\mathcal M}_0{\bf
q}^{\top}})^{\top}$ and equation (\ref{sys}) reads:
\begin{equation}\label{mk}
\left.\begin{array}{c} \alpha_3{\bf q}_{t_3}={\bf q}_{xxx}-3{\bf q}_x{\mathcal M}_0{\bf q}^{\top}{\bf q}_x+3\beta u{\bf q}_x,\\
\alpha_3 u_{t_3}=u_{xxx}-3(u{\bf q}_x{\mathcal M}_0{\bf
q}^{\top})_x+6\beta u u_x.
\end{array}\right.
\end{equation}
\end{enumerate}

 \item Let us put ${\mathcal M}_0=0$ in operators $L_1$, $M_3$
(\ref{LM}):

$L_1=D+\beta D^{-1} u+v$,
$M_3=\alpha_3\partial_{t_3}-D^3-3vD^2-3(v^2+v_x+\beta u)D$. Then
equation (\ref{S}) takes the form:
\begin{equation}\label{sys2}
\left.
\begin{array}{c}
-\alpha_3 v_{t_3}+v_{xxx}+3vv_{xx}+3v^2v_x+3v_x^2+6\beta(uv)_x=0,
\\
\alpha_3 u_{t_3}=u_{xxx}-3(uv)_{xx}+3(u(v^2+v_x+\beta
u))_x.\end{array} \right.
\end{equation}

\begin{enumerate}
\item Under additional restrictions $v=-2i{\rm Im}(D^{-1}\{\beta
u\})$ ($L_1^*=-DL_1D^{-1},\,\, M_3^*=-DM_3D^{-1}$) in item 2 we
obtain a complex generalization of the modified Korteweg-de Vries
equation:
\begin{equation}\label{sys22}
\begin{array}{l}
\alpha_3 u_{t_3}=u_{xxx}+6i(u Im(D^{-1}\{\beta
u\}))_{xx}+\\+3(u(-4Im(D^{-1}\{\beta u\})^2-2i{\rm Im}(\alpha
u)+\beta u))_x.
\end{array}
\end{equation}
In the real case ($\beta\in {\mathbb{R}}$, $u$  is a real-valued
function, $v=0$) operators $L_1$ and $M_3$ take the form:
$L_1=D+\beta D^{-1}u, M_3=\beta\partial_t-D^3-3\beta uD$, and we
obtain KdV equation in (\ref{sys22}):
\begin{equation}\label{kd}
\alpha_3 u_{t_3}=u_{xxx}+6\beta u u_x.
\end{equation}
\end{enumerate}
\item Let us put $u=0$ in Lax pair (\ref{LM}): $L_1=D-{\bf
q}{\mathcal M}_0D^{-1}{\bf r}^{\top}D+v,$
$M_3=\alpha_3\partial_{t_3}-D^3-3(v-{\bf q}{\mathcal M}_0{\bf
r}^{\top})D^2-3\left(({\bf q}{\mathcal M}_0{\bf r}^{\top}-v)^2-{\bf
q}{\mathcal M}_0{\bf r}^{\top}+v_x\right)D$. Equation (\ref{S})
becomes:
\begin{equation}\label{S23}
\left.
\begin{array}{c}
-\alpha_3 v_{t_3}+v_{xxx}+3vv_{xx}+3v^2v_x+3v_x^2+\\+3\left\{({\bf
q}{\mathcal M}_0{\bf r}^{\top})^2-{\bf q}_x{\mathcal M}_0{\bf
r}^{\top}\right\}v_x-3{\bf q}{\mathcal M}_0{\bf
r}^{\top}v_{xx}-6{\bf
q}{\mathcal M}_0{\bf r}^{\top}vv_x=0,\\
\alpha_3 {\bf q}_{t_3}={\bf q}_{xxx}+3(v-{\bf q}{\mathcal M}_0{\bf
r}^{\top}){\bf q}_{xx}+\\+3\left\{({\bf q}{\mathcal M}_0{\bf
r}^{\top}-v)^2-{\bf q}_x{\mathcal M}_0{\bf r}^{\top}+v_x\right\}{\bf q}_x,\\
\alpha_3{\bf r}^{\top}_{t_3}={\bf r}^{\top}_{xxx}- 3\left({\bf
r}^{\top}_x\left(v-{\bf q}{\mathcal M}_0{\bf
r}^{\top}\right)\right)_x+\\+3{\bf r}^{\top}_x\left\{({\bf
q}{\mathcal M}_0{\bf r}^{\top}-v)^2-{\bf q}_x{\mathcal M}_0{\bf
r}^{\top}+v_x\right\}.
\end{array}
\right.
\end{equation}
\end{enumerate}

\section{Dressing methods for k-cmKP hierarchy}

In this section our aim is to elaborate dressing methods for the
k-cmKP hierarchy (\ref{hier}). At first we recall a main result from
paper \cite{K2009}. Let $1\times K$-matrix functions $\varphi$ and
$\psi$ be solutions of linear problems with (2+1)-dimensional
generalization of the operator $L_k$ (\ref{eq21}) wit more general differential part $B_k$:
\begin{equation}\label{prx}
\begin{array}{c}
L_k\{\varphi\}=\varphi\Lambda,\,\,L_k^{\tau}\{\psi\}=\psi\tilde{\Lambda},\,\,\Lambda,\tilde{\Lambda}\in Mat_{K\times K}({\mathbb{C}}),\\
L_k=\beta_k\partial_{\tau_k}+B_k+{\bf q}{\mathcal M}_0D^{-1}{\bf
r}^{\top},\,\, B_{k}=\sum_{j=0}^{k}{u}_jD^j.
\end{array}
\end{equation}
Introduce a  binary Darboux transformation (BDT) in the following
way:
\begin{equation}\label{W}
W=I-\varphi\left(C+D^{-1}\{\psi^{\top}\varphi\}\right)^{-1}D^{-1}\psi^{\top}:=
I-\varphi\Delta^{-1}D^{-1}\psi^{\top},
\end{equation}
where $C$ is a $K\times K$-constant nondegenerate matrix. The
inverse operator $W^{-1}$ has the form:
\begin{equation}\label{W-}
W^{-1}=I+\varphi
D^{-1}\left(C+D^{-1}\{\psi^{\top}\varphi\}\right)^{-1}\psi^{\top}=I+\varphi D^{-1}\Delta^{-1}\psi^{\top}.
\end{equation}
 The following theorem is proven in \cite{K2009}.
\begin{theorem}{\cite{K2009}}\label{2009}
The operator $\hat{L}_k:=WL_kW^{-1}$  obtained from $L_k$ in
(\ref{prx}) via BDT (\ref{W}) has the form
\begin{equation}\label{Lop}
\begin{array}{l}
\hat{L}_k:=WL_kW^{-1}=\beta_k\partial_{\tau_k}+\hat{B}_k+\hat{\bf
q}{\mathcal M}_0D^{-1}{\hat{{\bf r}}}^{\top}+\Phi{\mathcal
M}D^{-1}\Psi^{\top},\,\\ \hat{B}_k=\sum_{j=0}^{k}\hat{u}_jD^j,
\end{array}
\end{equation}
where
\begin{equation}\label{DSM}
\begin{array}{l}
{\mathcal M}=C\Lambda-\tilde{\Lambda}^{\top}C,\,
\Phi=\varphi\Delta^{-1},\,\,
\Psi=\psi\Delta^{-1,\top},\,\Delta=C+D^{-1}\{\psi^{\top}\varphi\},\\
{\hat{\bf q}}=W\{{\bf q}\},\,\,{\hat{\bf r}}=W^{-1,\tau}\{{\bf r}\}.
\end{array}
\end{equation}
 $\hat{u}_j$ are scalar coefficients depending on functions $\varphi$,
$\psi$ and $u_i,\,\, i=\overline{0,j}$. In particular,
\begin{equation}\nonumber
\hat{u}_k=u_k,\quad
\hat{u}_{k-1}=u_{k-1},\ldots.
\end{equation}
The exact forms of all the coefficients $\hat{u}_j$ can be found in \cite{K2009}.
\end{theorem}
Using the previous theorem we obtain the following result for
(2+1)-generalization of operator $L_k$ from the k-cmKP hierarchy
(\ref{hier}):
\begin{theorem}\label{Theor}
Let $(1\times K)$-vector functions $\varphi$ and $\psi$ satisfy
linear problems:
\begin{equation}\label{pr}
\begin{array}{c}
L_k\{\varphi\}=\varphi\Lambda,\,\,L_k^{\tau}\{\psi\}=\psi\tilde{\Lambda},\,\,\Lambda,\tilde{\Lambda}\in Mat_{K\times K}({\mathbb{C}}),\\
L_k=\beta_k\partial_{\tau_k}+B_k-{\bf q}{\mathcal M}_0D^{-1}{\bf
r}^{\top}D,\,\, B_k=\sum_{i=1}^ku_iD^i.
\end{array}
\end{equation}
Then the operator $\hat{L}_k$ transformed via operator
\begin{equation}\label{Wm}
W_m:=w_0^{-1}W=w_0^{-1}\left(I-\varphi\Delta^{-1}D^{-1}\psi^{\top}\right)=I-\varphi{\tilde{\Delta}}^{-1}D^{-1}(D^{-1}\{\psi\})^{\top}D,\end{equation}
where
\begin{equation}\nonumber
\begin{array}{l}
w_0=I-\varphi\Delta^{-1}D^{-1}\{\psi^{\top}\},
\tilde{\Delta}=-C+D^{-1}\{D^{-1}\{\psi^{\top}\}\varphi_x\},\\\Delta=C+D^{-1}\{\psi^{\top}\varphi\},
\end{array}
\end{equation}

has the form:

\begin{equation}\label{Lop3}
\begin{array}{l}
\tilde{L}_k:=W_mL_kW_m^{-1}=\beta_k\partial_{\tau_k}+\tilde{B}_k-{\tilde{\bf
q}}{\mathcal M}_0D^{-1}{\tilde{{\bf
r}}}^{\top}D+{\tilde{\Phi}}{\mathcal
M}D^{-1}{\tilde{\Psi}}^{\top}D,\,\\
\tilde{B}_k=\sum_{j=1}^{k}\tilde{u}_jD^j,\,\,
\tilde{u}_k=u_k,\,\,\tilde{u}_{k-1}=u_{k-1}+ku_kw^{-1}_0w_{0,x},\ldots,
\end{array}
\end{equation}
where
\begin{equation}\label{DSM2}
\begin{array}{l}
{\mathcal M}=C\Lambda-\tilde{\Lambda}^{\top}C,\,
{\tilde{\Phi}}=-W_m\{\varphi\}C^{-1}=\varphi
{\tilde{\Delta}}^{-1},\,\,
\\{\tilde\Psi}=D^{-1}\{W_m^{\tau,-1}\{\psi\}\}C^{-1,\top}=D^{-1}\{\psi\}\Delta^{-1,\top},\,
{\tilde{\bf q}}=W_m\{{\bf q}\},\,\,\\{\tilde{\bf
r}}=D^{-1}W_m^{-1,\tau}D\{{\bf
r}\},\tilde{\Delta}=-C+D^{-1}\{D^{-1}\{\psi^{\top}\}\varphi_x\}.
\end{array}
\end{equation}
\end{theorem}
\begin{proof}
Let us check that
\begin{equation}\nonumber
w_0^{-1}=I-\varphi{\tilde{\Delta}}D^{-1}\{\psi^{\top}\},\,\,
\tilde{\Delta}=-C+D^{-1}\{D^{-1}\{\psi^{\top}\}\varphi_x\}.
\end{equation}
In order to do that we have to verify the equality $w_0w_0^{-1}=I$:
\begin{equation}\nonumber
w_0w_0^{-1}=I-\varphi\Delta^{-1}D^{-1}\{\psi^{\top}\}-\varphi{\tilde{\Delta}}^{-1}D^{-1}\{\psi^{\top}\}
$$$$+\varphi{\tilde{\Delta}}^{-1}\left(C+D^{-1}\{\psi^{\top}\varphi\}-C+D^{-1}\{D^{-1}\{\psi^{\top}\}\varphi_x\}\right)\varphi\Delta^{-1}D^{-1}\{\psi^{\top}\}=I\end{equation}
Analogously can be verified that $w_0^{-1}w_0=I$. By Theorem 1 we
obtain:
\begin{equation}\label{nuub}
\begin{array}{l}
W_mL_kW_m^{-1}=w_0^{-1}W\left(\beta_k\partial_{\tau_k}+B_k-{\bf
q}{\mathcal M}_0{\bf r}^{\top}+{\bf q}{\mathcal M}_0D^{-1}{\bf
r}_x^{\top}\right)W^{-1}w_0=\\
\beta_k\partial_{\tau_k}+(W_mL_kW_m^{-1})_{\geq0}-w_0^{-1}W\{{\bf
q}\}{\mathcal M}_0D^{-1}\left(W^{-1,\tau}\{{\bf
r}_x\}\right)^{\top}w_0+\\+w_0^{-1}\Phi{\mathcal
M}D^{-1}\Psi^{\top}w_0
\end{array}
\end{equation}
We shall point out that:
$\Psi^{\top}w_0=\Delta^{-1}\psi^{\top}(I-\varphi\Delta^{-1}D^{-1}\{\psi^{\top}\})=\left(\Delta^{-1}D^{-1}\{\psi^{\top}\}\right)_x=
{\tilde{\Psi}}^{\top}_{x}$. We shall also observe that:
\begin{equation}\nonumber
\begin{array}{l}
(W^{-1,\tau}\{{\bf r}_x\})^{\top}w_0=\left({\bf
r}^{\top}_x-D^{-1}\{{\bf r}^{\top}_x\varphi\}\Delta^{-1}\psi^{\top}
\right)\left(I-\varphi\Delta^{-1}D^{-1}\{\psi^{\top}\}\right)=\\
\left({\bf r}^{\top}-D^{-1}\{{\bf
r}^{\top}_x\varphi\}\Delta^{-1}D^{-1}\{\psi^{\top}\}\right)_x=(D^{-1}W_m^{-1,\tau}D\{{\bf
r}\})^{\top}_x=\tilde{{\bf r}}^{\top}_{x}
\end{array}
\end{equation}
Thus (\ref{nuub}) can be rewritten as:
\begin{equation}\label{AS}
\begin{array}{l}
\tilde{L}_k\!=\!W_mL_kW_m^{-1}\!=w_0^{-1}W\left(\!\beta_k\partial_{\tau_k}\!+B_k\!-{\bf
q}{\mathcal M}_0{\bf r}^{\top}+{\bf q}{\mathcal M}D^{-1}{\bf
r}_x^{\top}\right)W^{-1}w_0=\\
\beta_k\partial_{\tau_k}+(W_mL_kW_m^{-1})_{\geq0}+\tilde{{\bf
q}}{\mathcal M}_0D^{-1}{\tilde{\bf
r}}_{x}^{\top}-{\tilde\Phi}{\mathcal
M}D^{-1}{\tilde{\Psi}}_{x}^{\top}=\\
\beta_k\partial_{\tau_k}+(W_mL_kW_m^{-1})_{\geq0}+\tilde{{\bf
q}}{\mathcal M}_0{\tilde{\bf r}}^{\top}-{\tilde{\Phi}}{\mathcal
M}{\tilde{\Psi}}^{\top}-\tilde{{\bf q}}{\mathcal
M}_0D^{-1}{\tilde{\bf r}}^{\top}D-\\+{\tilde\Phi}{\mathcal
M}D^{-1}{\tilde{\Psi}}^{\top}D
\end{array}
\end{equation}
Using that $\tilde{L}_k\{1\}=\tilde{u}_0=0$ we obtain the form of
$\tilde{B}_k$. I.e.,
$\tilde{B}_k:=(W_mL_kW_m^{-1})_{\geq0}+\tilde{{\bf q}}{\mathcal
M}_0{\tilde{\bf r}}^{\top}-{\tilde{\Phi}}{\mathcal
M}{\tilde{\Psi}}^{\top}=\sum_{j=1}^k\tilde{u}_jD^j$.
\end{proof}
Theorem \ref{Theor} provides us with a dressing method for k-cmKP
hierarchy (\ref{hier}). I.e., the following corollary directly
follows from the previous theorem:
\begin{corollary}
Assume that operators $L_k$ and $M_n$ in (\ref{hier}) satisfy Lax
equation: $[L_k,M_n]=0$. Let functions $\varphi$ and $\psi$ satisfy
equations:\begin{equation}\label{pr2}
\begin{array}{c}
L_k\{\varphi\}=\varphi\Lambda,\,\,L_k^{\tau}\{\psi\}=\psi\tilde{\Lambda},\,\,\Lambda,\tilde{\Lambda}\in Mat_{K\times K}({\mathbb{C}}),\\
M_n\{\varphi\}=0,\,\, M_n^{\tau}\{\psi\}=0.
\end{array}
\end{equation}
Then transformed operators $\tilde{L}_k=W_mL_kW_m^{-1}$ (see
(\ref{Lop3}) with $\beta_k=0$) and
\begin{equation}\label{Mop}
\tilde{M}_n=W_mM_nW_m^{-1}=\alpha_n\partial_{t_n}-D^n-\sum_{i=1}^{n-1}\tilde{u}_iD^i
\end{equation}
via transformation $W_m$ (\ref{Wm}) also satisfy Lax equation:
$[\tilde{L}_k,\tilde{M}_n]=0$
\end{corollary}
\begin{proof}
It can be checked directly that:
$[\tilde{L}_k,\tilde{M}_n]=[W_mL_kW^{-1}_m,W_mM_nW_m^{-1}]=W_m[L_k,M_n]W_m^{-1}=0$.
The exact form of operators $\tilde{L}_k$ and $\tilde{M}_n$ follows
from Theorem \ref{Theor}.
\end{proof}
The following corollary follows from Corollary 1 and Theorem
\ref{Theor}:
\begin{corollary}
Suppose that functions $\varphi$ and $\psi$ satisfy equations
(\ref{pr2}) with operators $L_k$ and $M_n$ defined by (\ref{NH})
then transformed operators have the form:
\begin{equation}\label{NHL}
\begin{array}{l}
  \tilde{L}_k=B_k-\tilde{{\bf q}}{\mathcal M}_0D^{-1}\tilde{{\bf r}}^{\top}D+{\tilde{\Phi}}{\mathcal
M}D^{-1}{\tilde{\Psi}}^{\top}D+\tilde{v}+\beta
D^{-1}\tilde{u},\\
\tilde{M}_n=\alpha_n\partial_{t_n}-{\tilde{B}}_n,\,\,{\tilde{B}}_n=D^n+\sum_{i=1}^{n-1}\tilde{u}_iD^i,
\end{array}
\end{equation}
where

\begin{equation}\label{SM2}
\begin{array}{l}
{\mathcal M}=C\Lambda-\tilde{\Lambda}^{\top}C,\,
{\tilde{\Phi}}=-W_m\{\varphi\}C^{-1}=\varphi
{\tilde{\Delta}}^{-1},\,\,
\\{\tilde{\Psi}}=D^{-1}\{W_m^{\tau,-1}\{\psi\}\}C^{-1,\top}=D^{-1}\{\psi\}\Delta^{-1,\top},\,
{\tilde{\bf q}}=W_m\{{\bf q}\},\,\,\\{\tilde{\bf
r}}=W_m^{-1,\tau}\{{\bf
r}\},\tilde{\Delta}=-C+D^{-1}\{D^{-1}\{\psi^{\top}\}\varphi_x\},\Delta=C+D^{-1}\{\psi^{\top}\varphi\},\,\,\\
\tilde{u}=W_m^{-1,\tau}\{D^{-1}\{u\}\},\,\,\tilde{v}=W_m\{v\}+\beta
D^{-1}W_m^{-1,\tau}\{u\}-\beta W_m\{D^{-1}\{u\}\}.
\end{array}
\end{equation}
\end{corollary}
As it was shown in previous Sections the most interesting systems
arise from the k-cmKP hierarchy (\ref{hier}) and its reduction
(\ref{NH}) after a Hermitian conjugation reduction. Our aim is to
show that under additional restrictions Binary Darboux
Transformation $W_m$ (\ref{Wm}) preserves this reduction.

\begin{proposition}\label{her}
\begin{enumerate}
\item
Let $\psi=\bar{\varphi}_x$ and $C=-C^*$ in the dressing operator
$W_m$ (\ref{Wm}). Then the operator $W_m$ is D-unitary
($W_m^{-1}=D^{-1}W_m^*D$).
\item
Let the operator $L_k$ (\ref{hier}) be D-Hermitian:
$L_k^*=DL_kD^{-1}$ (D-skew-Hermitian: $L_k^*=-DL_kD^{-1}$) and
$M_{n}$ (\ref{hier}) be D-Hermitian (D-skew-Hermitian). Then the
operator $\hat{L}_k=W_mL_kW_m^{-1}$ (see (\ref{Lop3})) transformed
via the D-unitary operator $W_m$ is D-Hermitian (D-skew-Hermitian)
and $\hat{M}_{n}:=W_mM_{n}W_m^{-1}$ (\ref{Mop}) is D-Hermitian
(D-skew-Hermitian).
\item Assume that the conditions of items 1 and 2 hold. Let
$\tilde{\Lambda}=\bar{\Lambda}$ in the case of D-Hermitian $L_k$
($\tilde{\Lambda}=-\bar{\Lambda}$ in D-skew-Hermitian case). We
shall also assume that the function $\varphi$ satisfies the
corresponding equations in formulae (\ref{pr2}). Then ${\mathcal
M}={\mathcal M}^*$ (${\mathcal M}=-{\mathcal M}^*$) and
${\tilde{\Psi}}=
\bar{{\tilde{\Phi}}}$ (see formulae (\ref{DSM})).
\end{enumerate}
\end{proposition}
In subparagraph 4.1  we will show how one can use methods described
in Theorem \ref{Theor} and its corollaries in order to obtain
solutions of KdV equation (\ref{kd}) and its generalization
(\ref{mk})

\subsection{Solution generating technique for system (\ref{mk}) and KdV equation (\ref{kd}).}
We shall consider equation (\ref{mk}) in case the dimension of
vector ${\bf q}$ and matrix ${\mathcal{M}}_0$ is even. I.e.,
$m=2{\tilde{m}}$, ${\tilde{m}}\in{\mathbb{N}}$ (in this situation
skew-symmetric matrix ${\mathcal{M}}_0$ can be non-degenerate).
Assume that the skew-symmetric matrix ${\mathcal{M}}_0$ in
(\ref{mk}) and vector-function ${\bf{q}}$ has the form:
\begin{equation}\label{MMMM}
{\mathcal{M}}_0=\left(\begin{array}{cc}0_{\tilde{m}}&I_{\tilde{m}}\\ -I_{\tilde{m}}& 0_{\tilde{m}}\end{array}\right),\,\,
{\bf q}=({\bf q}_1,{\bf q}_2)=\left(q_{11},q_{12},\ldots, q_{1\tilde{m}},q_{21},q_{22},
\ldots, q_{2{\tilde{m}}}\right),
\end{equation}
where $0_{\tilde{m}}$ is a $\tilde{m}\times{\tilde{m}}$-dimensional
matrix consisting of zeros, $I_{\tilde{m}}$ is an identity matrix
with the dimension $\tilde{m}\times{\tilde{m}}$. Equation (\ref{mk})
after notation ${\tilde{u}}:=u$ can be rewritten in the following
form:
\begin{equation}\label{3c}
\left.\begin{array}{c}
\alpha_3{\bf q}_{1,t_3}={\bf q}_{1,xxx}-3({\bf q}_{1,x}{\bf q}^{\top}_2-{\bf q}_{2,x}{\bf q}^{\top}_{1}){\bf q}_{1,x}+3\beta {\tilde{u}}{\bf q}_{1,x},\\
\alpha_3{\bf q}_{2,t_3}={\bf q}_{2,xxx}-3({\bf q}_{1,x}{\bf q}^{\top}_2-{\bf q}_{2,x}{\bf q}^{\top}_{1}){\bf q}_{2,x}+3\beta {\tilde{u}}{\bf q}_{2,x},\\
\alpha_3 {\tilde{u}}_{t_3}={\tilde{u}}_{xxx}-3({\tilde{u}}({\bf q}_{1,x}{\bf q}^{\top}_2-{\bf q}_{2,x}{\bf q}^{\top}_{1}))_x+6\beta {\tilde{u}} {\tilde{u}}_x.
\end{array}\right.
\end{equation}
In this subsection our aim is to consider the case ${\tilde{m}}=1$ (although the corresponding solution generating technique
can be generalized to the case of an arbitrary natural ${\tilde{m}}$). In this situation ${\bf q}_1=q_1$ and ${\bf q}_2=q_2$ are scalars.
We shall suppose that $K=2{\tilde{K}}$ is an even natural number.
Assume that the function $\varphi$ is $(1\times K)$-vector solution of
the system:
\begin{equation}\label{thersys}
\left.
\begin{array}{c}
L_{10}\{\varphi\}=\varphi_x+\beta
D^{-1}\{u\varphi\}=\varphi\Lambda,\,\,\Lambda\in Mat_{{K}\times
{K}}(\mathbb{C}),\,\, \beta\in{\mathbb{R}},\\
M_{30}\{\varphi\}=\alpha_3\varphi_{t_3}-\varphi_{xxx}-3\beta
u\varphi_x=0,
\end{array}
 \right.
\end{equation}
with a number $u\in{\mathbb{R}}$.

Using Theorem \ref{Theor} and Proposition \ref{her} we obtain that
dressed operators $\tilde{L}_{10}$ and $\tilde{M}_{30}$ via operator
$W_m$ (\ref{Wm}) with skew-Hermitian matrix $C$ and
$\psi=\bar{\varphi}_x$ has the form:
\begin{equation}\label{DDS}
\begin{array}{l}
\tilde{L}_{10}=W_mL_{10}W_m^{-1}=D+{\tilde{\Phi}}{\mathcal
M}D^{-1}{\tilde{\Phi}}^*D+\beta
D^{-1}\tilde{u}+\tilde{v}\\
\tilde{M}_{30}=W_mM_{30}W_m^{-1}=\alpha_3\partial_{t_3}-D^3-(\tilde{v}+{\tilde{\Phi}}{\mathcal
M}{\tilde{\Phi}}^{*})D^2-\\-3\left(({\tilde{\Phi}}{\mathcal
M}{\tilde{\Phi}}^{*}+\tilde{v})^2+{\tilde{\Phi}}_{x}{\mathcal
M}{\tilde{\Phi}}^{*}+\tilde{v}_x+\beta {\tilde{u}}\right)D,
\end{array}
\end{equation}
where ${\mathcal{M}}=C\Lambda-\Lambda^*C^*$,
${\tilde{\Phi}}=\varphi{\tilde{\Delta}}^{-1}$,
$\tilde{u}=u-D\{\bar{\varphi}\bar{{\tilde{\Delta}}}^{-1}D^{-1}\{\varphi^{\top}u\}\}$,
$\tilde{v}=\beta({\tilde{\Phi}}
D^{-1}\{\varphi^*u\}-D^{-1}\{u\varphi\}{\tilde{\Phi}}^*)$,
$\tilde{\Delta}=-C+D^{-1}\{\varphi^*\varphi_x\}$. It has to be
pointed out that the function
${\tilde{\Phi}}=-W_m\{\varphi\}C^{-1}=\varphi{\tilde{\Delta}}^{-1}$
satisfies equation: $\tilde{M}_{30}\{{\tilde{\Phi}}\}=0$ because
$\tilde{M}_{30}\{{\tilde{\Phi}}\}=W_mM_{30}W_m^{-1}\{W_m\{\varphi\}C^{-1}\}=0$.

Now we assume that function $\varphi$, matrices $C$ and $\Lambda$
are real. In this case
$\tilde{v}=\tilde{v}^{\top}=\beta({\tilde{\Phi}}
D^{-1}\{\varphi^{\top}u\}-D^{-1}\{u\varphi\}{\tilde{\Phi}}^{\top})^{\top}=-{\tilde{v}}=0$.

Let us
put
\begin{equation}\label{matr}
\Lambda=diag(\lambda_{11},\lambda_{12},\lambda_{21},\lambda_{22},\ldots,\lambda_{{\tilde K}1},\lambda_{{\tilde K}2}),\,
\lambda_{ij}\in{\mathbb{R}},$$$$
C=\left(\begin{array}{ccccccc}C_{11}&C_{12}&\ldots&C_{1{\tilde K}}\\
C_{21}&C_{22}&\ldots&C_{2{\tilde K}}\\\vdots&\vdots&\ldots&\vdots\\C_{{\tilde K}1}&C_{{\tilde K}2}&\ldots&C_{{\tilde K}{\tilde K}}\end{array}\right),
\end{equation}
where elements $C_{ij}$ are $(2\times 2)$-matrices of the form:
\begin{equation}\label{Cmat}
C_{ij}=\left(\begin{array}{cc}0&-\frac{1}{\lambda_{j2}+\lambda_{i1}}\\
\frac{1}{\lambda_{j1}+\lambda_{i2}}&0\end{array}\right).
\end{equation}

Under such a choice of $C$ (\ref{Cmat}) and $\Lambda$ (\ref{matr})
we obtain that $2{\tilde K}\times 2{\tilde K}$-dimensional matrix ${\mathcal
M}=C\Lambda-\Lambda^{\top}C^{\top}$ has the block form: ${\mathcal
M}=({\mathcal M}_{ij})_{i,j=1}^{\tilde K}$, where ${\mathcal M}_{ij}={\mathcal M}_0$ (see formula (\ref{MMMM}) in case ${\tilde{m}}=1$).
Let us denote by: ${\bf 1}_{\tilde K}=(I_2,\ldots,I_2)$ matrix that consists
of ${\tilde K}$ $(2\times2)$-dimensional identity matrices $I_2$. Then
${\mathcal M}=-{\bf 1}_{\tilde K}^{\top}{\mathcal M}_0{\bf 1}_{\tilde K}$.

 Let us put $u=const$  and choose solution of system (\ref{thersys}) in the form:
$\varphi=\left(\varphi_{11},\varphi_{12}\varphi_{21},\varphi_{22},\ldots,\varphi_{{\tilde K}1},\varphi_{{\tilde K}2}\right)$,
$\varphi_{ij}=exp\left\{(\frac12\lambda_{ij}+\gamma_{ij})x+a_{ij}t\right\}$,
where $\gamma_{ij}=\sqrt{\frac14\lambda_{ij}^2-\beta u}$,
$a_{ij}=\left\{\left(\frac12\lambda_{ij}+\gamma_{ij}\right)^3+3\beta
u\left(\frac12 \lambda_{ij}+\gamma_{ij}\right)\right\}/\alpha_3$.
$(2{\tilde K}\times 2{\tilde K})$-matrix ${\tilde{\Delta}}$ then takes the block form:
\begin{equation}\label{sc}
{\tilde{\Delta}}=\!-C\!+\!D^{-1}\{\varphi^{\top}\varphi_x\}=\left({\tilde{\Delta}}_{ij}\right)_{i,j=1}^{\tilde K}=$$$$=\left(\begin{array}{cc}\!\!\!\frac{\alpha_{i1}}{\alpha_{i1}+\alpha_{j1}}e^{(\alpha_{i1}+\alpha_{j1})x+(a_{i1}+a_{j1})t}&
\frac{\alpha_{i2}}{\alpha_{i2}+\alpha_{j1}}e^{(\alpha_{i2}+\alpha_{j1})x+(a_{i2}+a_{j1})t}+\frac{1}{\lambda_{j2}+\lambda_{i1}}\\
\frac{\alpha_{i1}}{\alpha_{i1}+\alpha_{j2}}e^{(\alpha_{i1}+\alpha_{j2})x+(a_{i1}+a_{j2})t}-\frac{1}{\lambda_{j1}+\lambda_{i2}}&\frac{\alpha_{i2}}{\alpha_{i2}+\alpha_{j2}}e^{(\alpha_{i2}+\alpha_{j2})x+(a_{i2}+a_{j2})t}\end{array}\right)_{i,j=1}^{\tilde K},
\end{equation}
where $\alpha_{ij}=\frac12\lambda_{ij}+\gamma_{ij}$. Functions ${\bf
q}=(q_1,q_2)=\varphi{\tilde{\Delta}}^{-1}{\bf 1}_{\tilde K}^{\top}$ and
$\tilde{u}=u-
D\left\{\varphi{\tilde{\Delta}}^{-1}D^{-1}\{\varphi^{\top}u\}\right\}$
will be solutions of system (\ref{3c}).

 We
shall point out that in case $\beta=0$, ${\tilde K}=1$, $\alpha_3=1$
we obtain the following solution of the real version of the
mKdV-type equation (equation (\ref{3c}) with $\tilde{u}=0)$:
\begin{equation}\nonumber
\begin{array}{l}
{\bf q}=(q_1,q_2),\quad
q_1=-\frac{2(\lambda_{11}+\lambda_{12})\varphi_{12}}{(\lambda_{11}-\lambda_{12})\varphi_{11}\varphi_{12}-2},\quad
q_2=\frac{2(\lambda_{11}+\lambda_{12})\varphi_{11}}{(\lambda_{11}-\lambda_{12})\varphi_{11}\varphi_{12}-2},\\
\varphi_{1j}=e^{\lambda_{1j}x+\lambda_{1j}^3t_3},\,\lambda_{1j}>0,\,
j=\overline{1,2}.\,
\end{array}
\end{equation}
It is also possible to choose other types of matrices $C$ and
$\Lambda$ in (\ref{matr}) and (\ref{Cmat}). In particular the
following remark holds:
\begin{remark}
In case ${\tilde{K}}=1$ vector of functions $\varphi=(\varphi_1,\varphi_2)$,
$\varphi_1=\cos(x\lambda_{12}+(3\lambda_{11}^2\lambda_{12}-\lambda_{12}^3)t+\frac{\pi}4)e^{x\lambda_{11}+(\lambda_{11}^3-3\lambda_{11}\lambda_{12}^2)t}$,
$\varphi_2=\sin(x\lambda_{12}+(3\lambda_{11}^2\lambda_{12}-\lambda_{12}^3)t+\frac{\pi}4)e^{x\lambda_{11}+(\lambda_{11}^3-3\lambda_{11}\lambda_{12}^2)t}$.
will be a solution of the system (\ref{thersys}) with $u=0$ and
$\Lambda=\left(\begin{array}{cc}\lambda_{11}&\lambda_{12}\\-\lambda_{12}&\lambda_{11}\end{array}\right)$.
The corresponding solution generating technique given by
(\ref{matr})-(\ref{sc}) in case ${\tilde{K}}=1$,
$C_{\tilde{K}}=C_1=\left(\begin{array}{cc}0&\frac{1}{2\lambda_{11}}\\-\frac{1}{2\lambda_{11}}&0\end{array}\right)$
gives us a solution of mKdV-type equation (\ref{3c}) with $\tilde{u}=0$
that coincides with a solution obtained in \cite{T3BKP}.
\end{remark}

Now we will consider solution generating technique for KdV
(\ref{kd}). For this purpose we assume that function $\varphi$,
matrices $\Lambda=diag(\Lambda_1,\ldots,\Lambda_{\tilde{K}})$ and
$C=diag(C_1\ldots,C_{\tilde{K}})$ are real and have the form:
\begin{equation}
\Lambda_j=\left(\begin{array}{cc}0&\lambda_j\\\lambda_j&0\end{array}\right),\,C_j=\left(\begin{array}{cc}0&-c_j\\c_j&0\end{array}\right).
\end{equation}
In this case we obtain that the matrix ${\mathcal
M}=C\Lambda-\Lambda^{\top} C^{\top}$ consists of zeros in
(\ref{DDS}).   Consider the following solution of system
(\ref{thersys}):
\begin{equation}\nonumber
\varphi=\left(\varphi_{11},\varphi_{12},\varphi_{21},\varphi_{22},\ldots,\varphi_{{\tilde{K}}1},\varphi_{{\tilde{K}}2}\right),
$$$$
\varphi_{j1}=e^{\gamma_j x+a_j
t}\cosh\left(\frac{\lambda_j}2x+b_jt\right),\,\,
\varphi_{j2}=e^{\gamma_j x+a_j
t}\sinh\left(\frac{\lambda_j}2x+b_jt\right),
\end{equation}
where $\gamma_j=\sqrt{\frac14\lambda_j^2-\beta u}$,
$a_j=\left(\gamma_j^3+\frac34\gamma_j\lambda_j^2+3\beta
u\gamma_j\right)/\alpha_3$,
$b_j=\left(3\gamma_j^2\frac{\lambda_j}{2}+\frac{\lambda_j^3}{8}+\frac32\beta
u\lambda_j\right)/\alpha_3$ and $\lambda_j$, $\alpha_3$, $\beta$,
$u$ $\in{\mathbb{R}}$. Thus, $\tilde{v}=0$ and we obtain Lax pair
for KdV equation in (\ref{DDS}): $\tilde{L}_{10}=D+\beta
D^{-1}{\tilde{u}}$,
$\tilde{M}_{30}=\alpha_3\partial_{t_3}-D^3-3\beta {\tilde{u}}D$.

Formula
\begin{equation}\label{eqs}
\tilde{u}=u-D\left\{\varphi\tilde{\Delta}^{-1}D^{-1}\{\varphi^{\top}u\}\right\}
:=u+\hat{u},
\end{equation}
gives us a finite density solution of equation (\ref{kd}). In
particular, if ${\tilde{K}}=1$ and $c_1=\frac18\frac{\lambda_1}{\gamma_1}$ we obtain the following solution:
\begin{equation}\label{solut}
\tilde{u}=u+\frac{2\gamma_1^2}{\beta\cosh^2\left(\gamma_1x+a_1t\right)}.
\end{equation}
Now we shall substitute $\tilde{u}$ (\ref{eqs}) in KdV equation
(\ref{kd}):
\begin{equation}\label{kd2}
\alpha_3 \hat{u}_{t_3}=\hat{u}_{xxx}+6\beta \hat{u} \hat{u}_x+6\beta
u\hat{u}_x.
\end{equation}
The corresponding pair of operators have the form: $L_1=D+\beta
D^{-1}(\hat{u}+u), M_3=\alpha_3\partial_{t_3}-D^3-3\beta
\hat{u}D-3\beta uD$. We have two ways to obtain soliton solutions
(that are rapidly decreasing at both infinities in contradistinction
to finite density solutions (\ref{eqs}) that tend to an arbitrary
real number $u$) for KdV from formula (\ref{eqs}):
\begin{enumerate}
\item By taking the limit $u\rightarrow0$ in (\ref{eqs})-(\ref{kd2}).
\item By making a change of the
independent variables: $\tilde{x}:=x+6\alpha_3^{-1}\beta ut_3$, $\tilde{t}_3:=t_3$ and $\hat{v}({\tilde{x}},\tilde{t}_3):=\hat{u}(x,t_3)$ in
equation (\ref{kd2}) and solutions (\ref{eqs})-(\ref{solut}). This change corresponds to the change of differential operators in the
Lax pair for equation (\ref{kd2}) consisting of operators $L_1$ and $M_3$: $\alpha_3\partial_{\tilde{t}_3}=\alpha_3\partial_{t_3}-3\beta uD$.
\end{enumerate}

\section{Conclusions}

In this paper we obtain new generalizations (\ref{NH}) of the
modified k-cKP (k-cmKP) hierarchy (\ref{hier}). The obtained
hierarchy also generalizes the BKP hierarchy \cite{TBKP,T2BKP,T3BKP}
which is the special case of the k-cmKP hierarchy. Dressing methods
elaborated via BDT-type operators (Section 4) give rise to exact
solutions of the integrable systems that hierarchy (\ref{NH})
contains. In particular, soliton solutions for generalization of
mKdV-type equation (\ref{3c}) and finite density solutions as well
as regular soliton solutions were constructed for the KdV equation
using the proposed dressing methods. This methods also allow to
obtain rational and singular multi-soliton solutions of the
corresponding nonlinear systems under the special choice of spectral
matrix $\Lambda$ in the linear system (\ref{thersys}). In order to
minimize the size of this article we do not include those results
here.
 We shall point out that the special case
of equation (\ref{3c}) (${\tilde{u}}=0$) and its solutions were
considered in \cite{T3BKP}. Generalizations (\ref{NH}) of the k-cmKP
hierarchy (\ref{hier}) together with different extensions of k-cKP
hierarchy (including (1+1) and (2+1)-BDk-cKP hierarchy
\cite{n,NEW,NEW2}) is a good basis for construction of other
hierarchies of nonlinear integrable equations with corresponding
dressing methods. In particular in our forthcoming papers we plan to
introduce (2+1)-BDk-cmKP hierarchy and investigate solution
generating technique for the corresponding integrable systems.
Consider as an example Lax pair from the (1+1)-BDk-cKP hierarchy
that was investigated in \cite{NEW}:

\begin{equation}\label{ssex2+1}
\begin{array}{l}
\!P_{1,1}\!=D+\!c_1
\!M_2\{{\bf q}\}\!{\mathcal M}_0D^{-1}{\bf r}^{\top}+c_1{\bf
q}{\mathcal M}_0D^{-1}(M_2^{\tau}\{{\bf
r}\})^{\top}+c_0{\bf q}{\mathcal M}_0D^{-1}{\bf r}^{\top}\!=\!\\
=D+c_1\left(\alpha_2{\bf q}_{t_2}{\mathcal M}_0D^{-1}{\bf
r}^{\top}-\alpha_2{\bf q}{\mathcal M}_0D^{-1}{\bf
r}_{t_2}^{\top}-{\bf q}_{xx}{\mathcal M}_0D^{-1}{\bf
r}^{\top}\right.-\\\left.-{\bf q}{\mathcal M}_0D^{-1}{\bf
r}^{\top}_{xx}-u{\bf q}{\mathcal M}_0D^{-1}{\bf r}^{\top}-{\bf
q}{\mathcal M}_0D^{-1}{\bf r}^{\top}u\right)+c_0{\bf q}{\mathcal
M}_0D^{-1}{\bf r}^{\top},\\ M_{2}=\alpha_2\partial_{t_2}-D^2-u.
\end{array}
\end{equation}
It was shown in \cite{NEW} that the Lax equation $[P_{1,1},M_2]=0$
in (\ref{ssex2+1}) is equivalent to the system:
\begin{equation}\label{PM}
[P_{1,1},M_2]_{\geq0}=0,\, c_1M_2^2\{{\bf q}\}+c_0M_2\{{\bf
q}\}=0,\,\,c_1(M_2^{\tau})^2\{{\bf r}\}+c_0M_2^{\tau}\{{\bf r}\}=0.
\end{equation}
that is equivalent to the generalization of the AKNS system. In case
$c_0=1$, $c_1=0$ we obtain AKNS system in (\ref{PM}):
\begin{eqnarray}\nonumber
&&\alpha_2{\bf q}_{t_2}-{\bf q}_{xx}-u{\bf q}=0,\,\, -\alpha_2{\bf
r}_{t_2}-{\bf r}_{xx}-u{\bf r}=0,\,\, u=2{\bf q}{\mathcal M}_0{\bf
r}^{\top}.\label{TS}
\end{eqnarray}
Assume that the scalar function $f$ satisfies equations
$P_{1,1}\{f\}=f\lambda$, $M_2\{f\}=0$. We shall introduce the
notations $\tilde{M}_2:=f^{-1}M_2f$, $\hat{M}_2:=D\tilde{M}_2D^{-1}$,
$\tilde{P}_{1,1}:=f^{-1}P_{1,1}f$, ${\tilde{\bf q}}:=f^{-1}{\bf q}$,
${\tilde{\bf r}}^{\top}:=D^{-1}\{{\bf r}^{\top}f\}$ and consider the
following gauge transformations
\begin{equation}\nonumber
\begin{array}{l}
\tilde{M}_2=f^{-1}M_2f=\alpha_2\partial_{t_2}-D^2-2\tilde{u}D,\,\,\tilde{u}=f^{-1}f_x,\\
\tilde{P}_{1,1,}=f^{-1}P_{1,1}f=D+f^{-1}f_x+c_1
f^{-1}M_2\{{\bf q}\}\!{\mathcal M}_0D^{-1}{\bf
r}^{\top}f+\\+c_1f^{-1}{\bf q}{\mathcal M}_0D^{-1}(M_2^{\tau}\{{\bf
r}\})^{\top}f+c_0f^{-1}{\bf q}{\mathcal M}_0D^{-1}{\bf r}^{\top}f=\\
=D-c_1
\tilde{M}_2\{{\tilde{\bf q}}\}\!{\mathcal M}_0D^{-1}{\tilde{\bf
r}}^{\top}D-c_1{\tilde{\bf q}}{\mathcal
M}_0D^{-1}(\hat{M}_2^{\tau}\{\tilde{\bf r}\})^{\top}D-c_0{\tilde{\bf
q}}{\mathcal M}_0D^{-1}{\tilde{\bf r}}^{\top}D.
\end{array}
\end{equation}

The equation $[\tilde{M}_2,\tilde{P}_{1,1}]=0$ is equivalent to the
following system:
\begin{equation}\label{PM2}
[{\tilde{P}}_{1,1},{\tilde{M}}_2]_{>0}=0,\,
c_1{\tilde{M}}_2^2\{{\tilde{{\bf
q}}}\}+c_0{\tilde{M}}_2\{{\tilde{\bf
q}}\}=0,\,\,c_1(\hat{M}_2^{\tau})^2\{{\tilde{{\bf
r}}}\}+c_0\hat{M}_2^{\tau}\{{\tilde{{\bf r}}}\}=0.
\end{equation}
or in the equivalent form (after notation ${\bf q}_0:={\tilde{\bf
q}}$, ${\bf r}_0:={\tilde{\bf r}}$ ):
\begin{equation}\label{PM23}
\begin{array}{l}
[{\tilde{P}}_{1,1},{\tilde{M}}_2]_{>0}=0,\, {\bf
q}_1={\tilde{M}}_2\{{{{\bf q}_0}}\},\,\, {\bf
r}_1=\hat{M}_2^{\tau}\{{{\bf r}}_0\},
\\
c_1{\tilde{M}}_2\{{{{\bf q}_1}}\}+c_0{\tilde{M}}_2\{{{\bf
q}}_0\}=0,\,\,c_1\hat{M}_2^{\tau}\{{{\bf
r}}_1\}+c_0\hat{M}_2^{\tau}\{{{{\bf r}_0}}\}=0.
\end{array}
\end{equation}

 System (\ref{PM23}) is the generalization
of the Chen-Lee-Liu system (case $c_1=0$, $c_0=1$). In case of
additional reduction $\alpha_2\in i{\mathbb{R}}$, $c_0=0$, $c_1$$\in{\mathbb{R}}$, ${\mathcal
M}^*_0=-{\mathcal M}_0$, ${\bf r}=\bar{\bf q}$ (\ref{PM23}) reads as
following:
\begin{equation}
\begin{array}{l}
\alpha_2{\bf q}_{0,t_2}-{\bf q}_{0,xx}+2c_1({\bf
q}_1{\mathcal{M}}_0{\bf q}^{*}_0+{\bf q}_0\!{\mathcal{M}}_0\!{\bf
q}^{*}_1){\bf q}_{0,x}-{\bf q}_1=0,\,\,\\ \alpha_2{\bf
q}_{1,t_2}-{\bf q}_{1,xx}+2c_1({\bf q}_1{\mathcal{M}}_0{\bf
q}^{*}_0+{\bf q}_0{\mathcal{M}}_0\!{\bf q}^{*}_1){\bf q}_{1,x}=0.
\end{array}
\end{equation}

We shall also point out that the extension of the k-cmKP hierarchy
(\ref{NH}) can also be generalized to the matrix case. It leads to
matrix generalizations of integrable systems that hierarchy
(\ref{NH}) contains (including Chen-Lee-Liu (\ref{CLLiu}) and
modified-type KdV equation (\ref{kkdv})). In particular, the matrix
generalization of the modified KdV-type equation (\ref{kkdv})
differs from the well-known matrix mKdV equation  that was
investigated by the inverse scattering method in \cite{Khr}.


\section{Acknowledgment}

The second-named author Yu.M. Sydorenko (J. Sidorenko till 1998 in
earlier transliteration) thanks the Ministry of Education, Science,
Youth and Sports of Ukraine for partial financial support (Research
Grant MA-107F).

\end{document}